\documentclass[a4paper, reprint, aps, pra, showpacs,showkeys, twocolumn, nofootinbib]{revtex4-1}
\usepackage{xr}
\usepackage{blindtext}
\usepackage[utf8x]{inputenc}
\usepackage[bottom]{footmisc}
\usepackage{lipsum}
\usepackage{amsmath}
\usepackage{amsfonts}
\usepackage{amsthm}
\usepackage{bbm}
\usepackage{hyperref}
\usepackage{amssymb}
\usepackage{mathtools}
\usepackage{mathrsfs}
\usepackage{bbold}
\usepackage{array}
\usepackage{graphicx}
\usepackage{calc}
\usepackage{accents}
\theoremstyle{definition}
\newtheorem{example}{Example}
\def\be{\begin{equation}}
\def\ee{\end{equation}}
\def\bea{\begin{eqnarray}}
\def\eea{\end{eqnarray}}
\def\ben{\begin{equation*}}
\def\een{\end{equation*}}
\def\bean{\begin{eqnarray*}}
\def\eean{\end{eqnarray*}}
\def\bma{\begin{mathletters}}
\def\ema{\end{mathletters}}
\def\bi{\begin{itemize}}
\def\ei{\end{itemize}}
\newtheorem{thm}{Theorem}

\newtheorem{cor}{Corollary}[thm]

\newcommand{\pk}[3]{\langle #1 \,  | \, #2 \, | #3  \rangle}
\newcommand{\C}[1]{\mathbb{C}^{#1}\otimes\mathbb{C}^{#1}}

\newcommand{\proj}[1]{\ket{#1}\bra{#1}}
\newcommand{\ket}[1]{ | \, #1  \rangle}
\newcommand{\bra}[1]{ \langle #1 \,  |}

\newcommand{\ketbra}[2]{|{#1}\rangle \langle {#2}|}

\newcommand{\braket}[2]{\langle #1 | #2 \rangle}

\newcommand{\h}{\mathcal{H}_{\s{A}}\otimes\mathcal{H}_{\s{B}}}
\newcommand{\ha}{\mathcal{H}_{A'}\otimes\mathcal{H}_{\s{B}}}

\newcommand{\s}[1]{\scriptscriptstyle #1}

\newcommand{\Ss}{\mathcal{S}}
\newcommand{\T}{\mathcal{T}}
\newcommand{\Tb}{\mathcal{T}_\bot}
\newcommand{\W}{W_\mathbf{i}}
\newcommand{\Hh}{H_\mathbf{i}}
\newcommand{\A}{A_\mathbf{i}}
\newcommand\un[1]{\underaccent{\vec}{#1}}
\newcommand\G[1]{\mathcal{G}(#1)}

\newcommand\rr[1]{{\rho^{\s{(#1)}}_{\s{AB}}}}

\newcommand\CC{\mathcal{C}}
\newcommand\ma[1]{\mathfrak{u}{(#1)}}
\newcommand\Ee[2]{E^{\s{(#1)}}_{\s{#2}}}
\newcommand\EE[1]{E^{\s{(#1)}}}

\begin{document}

\title{Framework for distinguishability of orthogonal bipartite states by one-way local operations and classical communication}
\author{Tanmay Singal}
\email{stanmay@imsc.res.in}
\affiliation{Optics \& Quantum Information Group, The Institute of Mathematical Sciences, CIT Campus, Taramani, Chennai, 600 113, India}

\begin{abstract}
In the topic of perfect local distinguishability of orthogonal multipartite quantum states, most results obtained so far pertain to bipartite systems whose subsystems are of specific dimensions. In contrast very few results for bipartite systems whose subsystems are of arbitrary dimensions, are known. This is because a rich variety of (algebraic or geometric) structure is exhibited by different sets of orthogonal states owing to which it is difficult to associate some common property underlying them all, i.e., a common property that would play a crucial role in the local distinguishability of these states. In this paper, I propose a framework for the distinguishability by one-way LOCC ($1$-LOCC) of sets of orthogonal bipartite states in a $d_A \otimes d_B$ bipartite system, where $d_A, d_B$ are the dimensions of both subsytems, labelled as $A$ and $B$. I show that if the $i$-th party (where $i=A,B$) can initiate a $1$-LOCC protocol to perfectly distinguish among a set of orthogonal bipartite 
states, then the information of the existence of such a $1$-LOCC protocol lies in a subspace of $d_i \times d_i$ hermitian matrices, denoted by  $\Tb^{(i)}$, and that the method to extract this information (of the existence of this $1$-LOCC protocol) from $\Tb^{(i)}$ depends on the value of $dim \Tb^{(i)}$. In this way one can give sweeping results for the $1$-LOCC (in)distinguishability of all sets of orthogonal bipartite states corresponding to certain values of $dim \Tb^{(i)}$. Thus I propose that the value of $dim \Tb^{(i)}$ gives the common underlying property based on which sweeping results for the $1$-LOCC (in)distinguishability of orthogonal bipartite quantum states can be made. 
\end{abstract}



\maketitle

\label{intro}

\textbf{Introduction:} The scenario in local distinguishability of bipartite orthogonal quantum states is as follows: Alice and Bob are given one of many possible orthogonal bipartite states and they have to figure out which one they've been given using only local operations and classical communication (LOCC). Some prominent results which apply to joint systems, whose subsystems are of arbitrary dimension, are Bennet et al's result \cite{B99}, which established that members from an unextendible product basis cannot be perfectly distinguished by LOCC, Walgate et al's result \cite{W00}, which establishes that any two multipartite orthogonal quantum states can be perfectly distinguished using only LOCC, Badziag et al's  \cite{B03} result, which 
obtained a Holevo-like upper bound for the locally accessible information for an ensemble of states from a bipartite system, and Cohen's result \cite{C08}, which established that almost all sets of $d+1$ orthogonal states from $N$ $d$-dimensional multipartite systems are not perfectly distinguishable by LOCC. Very few such generic results are known. In this paper I propose a framework for the one-way LOCC distinguishability of orthogonal bipartite states, and this proposition is made as an aforementioned generic result.

\textbf{Framework:} Inspired by work done in \cite{Y07} and \cite{N13}, I show that for a given set of orthogonal bipartite mixed states from a $d_A \otimes d_B$ bipartite system, the $i$-th party (where $i=A,B$) can be associated with a subspace of $d_i \times d_i$ hermitian matrices, $\Tb^{(i)}$ (defined after equation \eqref{T}) which contains all information of one-way LOCC ($1$-LOCC) protocols which this $i$-th party can initiate to perfectly distinguish among said given set of orthogonal bipartite states. In this paper I obtain results to extract this $1$-LOCC related information from $\Tb^{(i)}$. 

For simplifying notation, I make two assumptions, which won't reduce the generality of results obtained: (1) Alice always initiates the protocol. This allows for simplifying the notation: $\Tb^{(A)} \longrightarrow \Tb$. Note that to establish distinguishability by $1$-LOCC ($1$-LOCC distinguishability), one \emph{has} to extract relevant information from $\Tb^{(A)}$ (for Alice starting protocol) and/or $\Tb^{(B)}$ (for Bob starting protocol), separately. (2) If $d_A < d_B$, one can always extend Alice's subsystem $A$ to a larger local system $A'$ so that $d_{A'}=d_B$. Similarly, vice versa. Thus, there's no loss of generality in assuming that $d_A = d_B = d$.

Let Alice and Bob have $d$ dimensional quantum systems whose Hilbert spaces are denoted by $\mathcal{H}_{\s{A}}$ and $\mathcal{H}_{\s{B}}$ respectively. Let them share one of $n$ orthogonal bipartite states, whose density matrices $\rr{1},\rr{2},\cdots, \rr{n}$ are observables on $\mathcal{H}_{\s{A}}\otimes \mathcal{H}_{\s{B}}$. They wish to establish which state they share using a $1$-LOCC protocol which Alice commences. Let spectral decomposition of $\rr{i}$ be

\begin{equation}
\label{rho}
\rr{i} = \sum_{j=1}^{r_i} \lambda_{ij} \ketbra{\psi_{ij}}{\psi_{ij}},
\end{equation}

where $r_i = rank \left(\rr{i}\right)$, $\{ \lambda_{ij} \}_{j=1}^{r_i}$ are non-zero eigenvalues of $\rr{i}$ and $\braket{\psi_{ij}}{\psi_{i'j'}}= \delta_{ii'}\delta_{jj'}$, $\forall$ $1 \leq i \leq i' \leq n$, $1 \leq j \leq r_{i}$ and $ 1 \leq j' \leq r_{i'}$. Let $\{ \ket{s_{\s{j}}}_{\s{A}} \}_{j=1}^{d}$ and $\{ \ket{s_{\s{j}}}_{\s{B}} \}_{j=1}^{d}$ be standard orthonormal bases (ONB) for $\mathcal{H}_{\s{A}}$ and $\mathcal{H}_{\s{B}}$ respectively. For any $ 1 \leq i \leq n$, and any $1 \leq j \leq r_i$, define $d \times d$ complex matrices $W_{ij}$ by expanding $\ket{\psi_{ij}}_{\s{AB}}$ in the basis $\{\ket{s_j}_{\s{A}}\ket{s_k}_{\s{B}}\}_{j,k=1}^{d}$

\begin{equation}
\label{Wi}
\ket{\psi_{ij}}_{\s{AB}}= \sum_{k,l=1}^{d} \left( W_{ij} \right)_{kl}  \ket{s_{\s{l}}}_{\s{A}}\ket{s_{\s{k}}}_{\s{B}}.
\end{equation}

Thus $\braket{\psi_{ij}}{\psi_{i'j'}} =Tr(W_{ij}^\dag W_{i'j'}) =\delta_{\s{ii'}}\delta{\s{jj'}},$ $\forall$ $1 \leq i \le i' \leq n$, $1 \leq j \leq r_{i}$ and $ 1 \leq j' \leq r_{i'}$. Define the index set $\mathcal{I}$ $\equiv$ $\{ (i, i',j,j'),$ $ | $ $1 \leq i < i' \leq n,$ $1 \leq j \leq r_{i}$, $1 \leq j' \leq r_{i'}\}$. Let $\mathbf{i} = (i, i',j,j') \in \mathcal{I}$. Define $W_{\mathbf{i}} \equiv {W_{ij}}^\dag W_{i'j'}$. Then $W_{\mathbf{i}}$'s are $d \times d$ complex matrices with trace zero. Let $H_{\mathbf{i}}$ $\equiv$ $\frac{1}{2} \left(W_{\mathbf{i}} + \left(W_{\mathbf{i}}\right)^\dag\right)$ and $A_{\mathbf{i}}$ $\equiv$ $\frac{1}{2i} \left(W_{\mathbf{i}} - \left(W_{\mathbf{i}}\right)^\dag\right)$. Let $\Ss$ be the real vector space of all $d \times d$ hermitian matrices. $dim \Ss = 
d^2$. Let $\T$ be a subspace of $\Ss$, defined by
\begin{equation}
\label{T}
\T \equiv \left\{ \sum_{\s{ \mathbf{i}} \in \mathcal{I}} a_{\s{\mathbf{i}}} H_{\s{\mathbf{i}}} + b_{\s{\mathbf{i}}} A_{\s{\mathbf{i}}}, \; \forall \;  a_{\s{\mathbf{i}}}, b_{\s{\mathbf{i}}} \in \mathbb{R}   \right\}.
\end{equation} 

Let $\Tb$ be the orthogonal complement of $\T$ in $\Ss$. Note that $\mathbb{1}_d$ $\in$ $\Tb$, where $\mathbb{1}_d$ is the $d \times d$ identity matrix. 

Now consider theorem \ref{Nat}.

\begin{thm}[Nathanson \cite{N13}]
 \label{Nat}
Alice can commence a $1$-LOCC protocol to distinguish among $\rr{1},\rr{2},\cdots,\rr{n}$ if and only if an orthogonality preserving (OP) rank-one POVM exists on her side to start protocol with.
\end{thm}

The set of POVMs acting on a quantum system $A$ (or $B$) is convex, and a rank-one POVM $\{ \proj{\tilde{l}} \}_{l=1}^{m}$ in that set isn't necessarily extremal (check supplemental material \cite{suppstructure} for more information). $\sum_{l=1}^{n} \proj{\tilde{l}} = \mathbb{1}_{\s{A}}$, where $\mathbb{1}_{\s{A}}$ is the identity operator on $\mathcal{H}_{\s{A}}$. Let $\{ \proj{\tilde{l}} \}_{l=1}^{m}$ have a  convex decomposition into two distinct extremal rank-one POVMs: $\proj{\tilde{l}} = p \proj{\tilde{l'}} + (1-p)\proj{\tilde{l''}}$, $\forall \; 1 \le l \le m$, and where $ p \in (0,1)$, $\sum_{l=1}^{m} \proj{\tilde{l'}} = \sum_{l=1}^{m} \proj{\tilde{l''}} = \mathbb{1}_{\s{A}}$. This is possible if and only if $\ket{\tilde{l'}}_{\s{A}}$ and $\ket{\tilde{l''}}_{\s{A}}$ are both scalar multiples of $\ket{\tilde{l}}_{\s{A}}$. Thus if $\{ \proj{\tilde{l}} \}_{l=1}^{m}$ is OP, then so are $\{ \proj{\tilde{l'}} \}_{l=1}^{m}$ and $\{ \proj{\tilde{l''}} \}_{l=1}^{m}$. Theorem \ref{Nat} implies that since $\{ \
proj{\tilde{l'}} \}_{l=1}^{m}$ (and $\{ \proj{\tilde{l''}} \}_{l=1}^{m}$) is OP, Alice can commence protocol with $\{ \proj{\tilde{l'}
} \}_{l=1}^{m}$ (or $\{ \proj{\tilde{l''}} \}_{l=1}^{m}$). Thus, \emph{if $\rr{1},\rr{2},\cdots,\rr{n}$ are $1$-LOCC distinguishable, Alice can always choose her starting measurement to be an extremal rank-one POVM}.

\begin{thm}
\label{thm1}
$\rr{1},\rr{2},\cdots,\rr{n}$ are $1$-LOCC distinguishable if and only if $\Tb$ contains all elements of an extremal rank-one POVM. 
\end{thm}
\begin{proof}
\textbf{ONLY IF:} Assume that $\rr{i}$'s are $1$-LOCC distinguishable. Thus there exists an OP extremal rank-one POVM $\{\proj{\tilde{l}}\}_{l=1}^{m}$ on Alice's side. Let Kraus operators of this measurement be $\{ \ketbra{\phi_l}{\tilde{l}} \}_{l=1}^{m}$, where $\ket{\phi_l}_{\s{A}}$ are normalized. If the measurement outcome is $k$, the (unnormalized) $i$-th post-measurement state is $\left( \ketbra{\phi_k}{\tilde{k}} \otimes \mathbb{1}_{\s{B}} \right)$ $\rr{i}$ $\left( \ketbra{\tilde{k}}{\phi_k} \otimes \mathbb{1}_{\s{B}} \right)$, where $\mathbb{1}_{\s{B}}$ is the identity operator acting on $\mathcal{H}_{\s{B}}$. Since the $k$-th POVM element is OP, we get the following equations for all $ 1 \leq i < i' \leq n$, $1 \leq j \leq r_{i}$ and $1 \le j' \le r_{i'}$.
\small
\begin{align}
\label{OP}
& Tr\left( \left( \proj{\tilde{k}} \otimes \mathbb{1}_{\s{B}} \right) \rr{i} \left( \proj{\tilde{k}} \otimes \mathbb{1}_{\s{B}} \right) \rr{i'} \right) = 0,  \\
\Longrightarrow \; \;&  \left(\rr{i}\right)^\frac{1}{2} \left( \proj{\tilde{k}} \otimes \mathbb{1}_{\s{B}} \right) \left( \rr{i'}\right)^\frac{1}{2} = 0, \notag \\
\label{OPeffect}
\Longrightarrow \; \; & \prescript{}{\s{AB}}{\pk{\psi_{ij}}{\left(\proj{\tilde{k}}\otimes \mathbb{1}_{\s{B}}\right)}{\psi_{i'j'}}_{\s{AB}}} =0.
\end{align}\normalsize 

Substituting expressions for $\ket{\psi_{ij}}_{\s{AB}}$ from equation \eqref{Wi} in equation \eqref{OPeffect} we get
\begin{align}
\label{in1}
 \sum_{b,b'=1}^{m} & \braket{s_b}{\tilde{k}} \left(W_\mathbf{i}\right)_{bb'}\braket{\tilde{k}}{s_{b'}} =0.
\end{align}

Since $\{ \proj{\tilde{l}} \}_{l=1}^{m}$ is a POVM, there exists an $m \times d$ isometry matrix $U$ such that $ \ket{\tilde{l}}_{\s{A}}=\sum_{l'=1}^{d} U_{l l'}\ket{s_{l'}}_{\s{A}}$. Using $U$, define the following $m$ vectors in $\mathbb{C}^d$: $ \ket{\tilde{l}^*} \equiv \left(U^{\s{*}}_{\s{l1}},U^{\s{*}}_{\s{l2}},\cdots,U^{\s{*}}_{\s{ld}}\right)^{\s{T}}$. Then $ \braket{\tilde{k}}{s_{b'}} = U^{\s{*}}_{\s{kb'}} $. Using this in equation \eqref{in1} implies that $ \pk{\tilde{k}^*}{\W}{\tilde{k}^*} = 0$ which implies that $\pk{\tilde{k}^*}{\Hh}{\tilde{k}^*}= \pk{\tilde{k}^*}{\A}{\tilde{k}^*} =0, \; \forall \; \mathbf{i} \in \mathcal{I}$. Thus $\proj{\tilde{k}^*} \in \Tb$. Similarly, $\{ \proj{\tilde{l}^*}\}_{l=1}^{m}$ is an extremal rank-one POVM contained in $\mathcal{T}_\bot$. \textbf{IF} Let $\{ \proj{\tilde{l}^*}\}_{l=1}^{m} \subset \Tb$ be an extremal rank-one POVM. It is readily seen that arguments presented in the ONLY IF part can be traced backwards to conclude that Alice has a corresponding extremal 
rank-one OP POVM of the form $\{ \proj{\tilde{l}} \}_{l=1}^{m}$.
\end{proof}

If Bob were to start protocol, one would have to check if $\Tb^{\s{(B)}}$ contains all elements of some rank-one POVM, instead of $\Tb^{\s{(A)}}$ (denoted by $\Tb$ here). Note that $\Tb^{\s{(B)}}$ is defined to be the complement of $\mathcal{T}^{\s{(B)}}$ in $\mathcal{S}$, where $\mathcal{T}^{\s{(B)}}$ is defined just such as $\mathcal{T}$ was in equation \eqref{T}, with the difference that $W_{\mathbf{i}} \equiv W_{ij}W_{i'j'}^\dag$, not $W_{ij}^\dag W_{i'j'}$. 

Any $d$ dimensional subspace of $\Ss$ is called a \emph{maximally abelian subspace} (MAS) if all matrices in it commute. Any MAS can be associated with a unique common eigenbasis such that all hermitian matrices, which are diagonal in said common eigenbasis, lie in the MAS. 

\begin{cor} 
\label{cor2} 
$\rr{1},\rr{2},\cdots,\rr{n}$ are $1$-LOCC distinguishable using only projective measurements on $\mathcal{H}_{\s{A}}$ and $\mathcal{H}_{\s{B}}$, if and only if $\Tb$ contains a MAS. 
\end{cor}
\begin{proof}
\textbf{ONLY IF}: Let $\rr{i}$'s be $1$-LOCC distinguishable using only projective measurements on $\mathcal{H}_{\s{A}}$ and $\mathcal{H}_{\s{B}}$. Thus Alice can initiate protocol by an OP rank-one projective measurement $\{ \proj{k} \}_{k=1}^{d}$. Then theorem  \ref{thm1} (ONLY IF part) implies that $\Tb$ contains all projectors of a rank-one projective measurement $\{ \proj{k^*}\}_{k=1}^{d}$. $span \left( \{ \proj{k^*}\}_{k=1}^{d} \right)$ is a MAS in $\Tb$. \textbf{IF}: Assume that $\Tb$ contains a MAS of $\Ss$. This MAS contains all matrices which are diagonal in MAS's common eigenbasis $\{ \ket{k^*}\}_{k=1}^{d}$. Thus this MAS contains the subset $\{\proj{k^*}\}_{k=1}^{d}$, which is a rank-one projective measurement. Then theorem (IF part) \ref{thm1} implies that $\rr{i}$'s are $1$-LOCC distinguishable by projective measurements. \end{proof}

The significance of corollary \ref{cor2} is that for certain values of $dim \T_\bot$, it is easy to check if $\T_\bot$ contains a MAS or not, which immediately indicates the existence or non-existence of a $1$-LOCC protocol (using only rank-one projective measurements). 

Non-existence of a MAS in $\T_\bot$ does not rule out the existence of a non-projective extremal rank-one POVM $\{\proj{\tilde{l}^*}\}_{l=1}^{m}$ in $\T_\bot$, where $m > d$. Theorem \ref{thm1} implies that if $\Tb$ contains $\{\proj{\tilde{l}^*}\}_{l=1}^{m}$, then there exists a $1$-LOCC distinguishability protocol which commences with an OP non-projective extremal rank-one POVM $\{\proj{\tilde{l}}\}_{l=1}^{m}$. Then one can consider $\mathcal{H}_{\s{A}}$ to be a $d$-dimensional subspace of an extended $m$-dimensional space $\mathcal{H}_{A'}$, so that $\ket{\psi_{\s{ij}}}_{\s{AB}}$ $\longrightarrow$ $\ket{\psi_{\s{ij}}}_{\s{A'B}}$ lie in $\ha$. Then $\Ss'$, $\T'$ and $\Tb'$ are spaces of $m \times m$ hermitian matrices corresponding to Alice's extended space $\mathcal{H}_{\s{A'}}$, and $\Tb'$ will contain an $m$-dimensional MAS, which corresponds to an $m$-element rank one projective measurement on $\mathcal{H}_{\s{A'}}$. This $m$-element projective measurement reduces to $\{\proj{\tilde{l}^*}\}_{l=1}^{m}$ 
when $\mathcal{H}_{\s{A'}}$ is limited to $\mathcal{H}_{\s{A}}$. Note that since POVM elements of any extremal rank-one POVM are linearly independent (LI) \cite{A05}, $m \le d^2$. It is sensible to search for an $m$-dimensional MAS in $\Tb'$ after confirming that $\Tb$ doesn't contain a $d$-dimensional MAS. Often the value of $dim \Tb$ itself gives information about OP rank-one POVMs which Alice can perform, e.g., Walgate et al's result \cite{W00}, that any two orthogonal bipartite pure states are $1$-LOCC distinguishable, which corresponds to the cases $dim \Tb \ge d^2 -2$. I give an alternative proof of Walgate et al's result in the supplemental material \cite{suppstructure}. Another example: when $dim \Tb = 1$, $\rr{i}$'s aren't distinguishable by LOCC at all \cite{Y07,C08}. For $1$-LOCC, corollary \ref{cor5} makes a stronger statement.

\begin{cor}
\label{cor5}
If $dim \Tb \le d-1$, there is no $1$-LOCC protocol which Alice can initiate to distinguish the states. 
\end{cor}

\begin{proof}
If $dim \Tb \le d-1$, $\Tb$ can't contain all $m$ ($\ge d$) LI elements of an extremal rank-one POVM. Then theorem \ref{thm1} implies that there is no such protocol. \end{proof}

\begin{cor}
\label{cor3}
When  $dim \Tb = d$, states are $1$-LOCC distinguishable if and only if $\Tb$ is a MAS of $\Ss$. 
\end{cor}
\begin{proof} \textbf{IF:} Already covered in corollary \ref{cor2}. \textbf{ONLY IF:} Given that $dim \Tb = d$ and the states are $1$-LOCC distinguishable. Theorem \ref{thm1} implies that $\Tb$ contains all POVM elements of an extremal rank-one POVM $\{ \proj{\tilde{k}^*} \}_{k=1}^{m}$. Since $dim \Tb = d$, elements of an extremal rank-one POVM being LI \cite{C04} implies that $m=d$. Thus the isometric matrix relating $\{ \ket{\tilde{k}^*} \}_{k=1}^{d}$ to an ONB of $\mathbb{C}^d$ has to be a $d \times d$ unitary matrix, which implies that $\{ \proj{\tilde{k}^*} \}_{k=1}^{d}$ $\longrightarrow$ $\{ \proj{k^*} \}_{k=1}^{d}$ is a rank-one projective measurement. Since $span \left(\{ \proj{k^*} \}_{k=1}^{d}\right)$ $=$ $\Tb$, $\Tb$ is a MAS of $\Ss$. 
\end{proof}

Consider the case when $n=d$ and the states are pure: $\rr{i} \longrightarrow \ket{\psi_{i}}_{\s{AB}}$. Then the index set $\mathcal{I}$ is $\left\{ (i,i'),\; \forall \; 1 \leq i < i' \leq d \right\}$. The cardinality of $\mathcal{I}$ now is $\frac{d(d-1)}{2}$. One can generally expect $\{ \Hh, \A\}_{\mathbf{i} \in \mathcal{I}}$ to be a LI set, which implies that $dim \T = d(d-1)$ and $dim \Tb = d$ for almost all sets of $d$ orthogonal states in $\h$. This is indeed true; proof for this was essentiallygiven by Cohen in \cite{C08}, where he showed that almost all sets of $n \ge d+1$ orthogonal multipartite qudit states in $d^{\otimes N}$ systems ($N \ge 2$) are locally indistinguishable, but for the sake of completeness I give a rigorous proof for this case in the supplemental material \cite{suppstructure}. Thus \emph{corollary \ref{cor3} gives the necessary and sufficient condition for the $1$-LOCC distinguishability of almost all sets of $d$ orthogonal pure states from $\h$}. Next, consider an example of 
this. 

\begin{example}
\label{exam2}
Define the following states in $\C{4}$: \small \begin{equation}
\label{GBS}
\ket{\psi_{\s{nm}}}_{\s{AB}} \equiv \sum_{j,k=0}^{3}\left( W_{nm} \right)_{kj}\ket{s_j}_{\s{A}}\ket{s_k}_{\s{B}},
\end{equation}
\normalsize where $\left(W_{\s{nm}}\right)_{\s{kj}} \equiv \frac{e^{\frac{ i \pi j  n }{2}}}{2}\delta_{j \oplus_4 m, k}$, $\forall$ $j,k = 0,1,2,3$. Note that any two $W_{nm}$ matrices are pairwise orthogonal.
For $1$-LOCC of the states $\{ \ket{\psi_{\s{00}}}_{\s{AB}}$, $\ket{\psi_{\s{01}}}_{\s{AB}}$, $\ket{\psi_{\s{10}}}_{\s{AB}}$, $\ket{\psi_{\s{33}}}_{\s{AB}}\}$, $\T$ is spanned by the hermitian matrices: \small$ \frac{W_{01}+W_{03}}{2}$, $\frac{W_{01}-W_{03}}{2i}$, $\frac{W_{10}+W_{30}}{2}$, $\frac{W_{10}-W_{30}}{2i}$, $\frac{W_{33}+iW_{11}}{2}$, $\frac{W_{01}-W_{03}}{2i}$, $\frac{W_{13}-iW_{31}}{2}$, $\frac{W_{13}+iW_{31}}{2i}$, $i\frac{W_{32}+W_{12}}{2}$, $\frac{W_{32}-W_{12}}{2}$, $\frac{W_{23}-W_{21}}{2}$,$\frac{W_{23}+W_{21}}{2i}$\normalsize. Hence $dim\T = 12$. Thus $dim \Tb = 4$, where $\Tb$ is spanned by the hermitian matrices $\mathbb{1}_4$, $ W_{22}$, $ W_{02}$, and $ W_{20}$. Note that all these matrices commute with each other. Thus $\Tb$ is a MAS. The common eigenbasis, which diagonalizes any matrix in $\Tb$ is $\big\{ \frac{1}{\sqrt{2}}(1,0,1,0)^T $, $\frac{1}{\sqrt{2}}(1,0,-1,0)^T$,$\frac{1}{\sqrt{2}}(0,1,0,1)^T$, $ \frac{1}{\sqrt{2}}(0,1,0,-1)^T \big\}$. It's then seen that Alice can initiate 
a $1$-LOCC protocol to distinguish the given set of states by performing rank-one projective measurement in the ONB $\big\{ \frac{\ket{0}{\s{A}}+\ket{2}_{\s{A}}}{\sqrt{2}} $, $\frac{\ket{0}{\s{A}}-\ket{2}_{\s{A}}}{\sqrt{2}}$, $\frac{\ket{1}{\s{A}}+\ket{3}_{\s{A}}}{\sqrt{2}}$, $\frac{\ket{1}{\s{A}}-\ket{3}_{\s{A}}}{\sqrt{2}}\big\}$.
\end{example}

So when $dim \Tb \le d$ or $dim \Tb \ge d^2-2$, one can conclude if $\Tb$ contains a MAS or not. When $ d+1 \leq dim \Tb \leq d^2-3$, it is difficult to establish the same, but one can give partial results. Let $dim \Tb = d + t$, where $ t \ge 1$. Let $\{ T_i \}_{i=1}^{d+t}$ be an ONB for $\Tb$. Let $\CC$ be the real vector space, spanned by matrices in $\{ i[ T_j,T_k ] \; $ $ | \; 1 \leq j < k \leq d+t  \}$, where $[ T_j,T_k ]$ $\equiv$ $T_j T_k$ $-$ $T_k T_j$. 

\begin{thm}
\label{thm2}
When $ 1 \le t  \leq \sqrt{3 d^2 - 3 d + \frac{1}{4}} - (d-\frac{3}{2})$ , $\Tb$ contains no MAS if $dim \CC$ $>$ $td + \frac{t(t-3)}{2}$.
\end{thm}
\begin{proof} If $\Tb$ contains a MAS, choose $\{ T_i \}_{i=1}^{d+t}$ such that $\{ T_i \}_{i=1}^{d}$ is an ONB for this MAS. Then number of non-zero commutators in $\{ i[ T_j,T_k ] \; $ $ | \; 1 \leq j < k \leq d+t  \}$ is at most $td + \frac{t(t-3)}{2}$, implying that $dim \CC$ can be at most $Min\left\{ td + \frac{t(t-3)}{2},d^2-1\right\}$. When $ 1 \leq t  \leq \sqrt{3 d^2 - 3 d + \frac{1}{4}} - (d-\frac{3}{2})$ then $d \le td + \frac{t(t-3)}{2}$ $\leq$ $d^2-1$. Then if $dim \CC$ $>$ $td + \frac{t(t-3)}{2}$, $\Tb$ contains no MAS. 
\end{proof}

For $dim \Tb = d+1$, I give necessary and sufficient conditions for $\Tb$ to contain a MAS. Let $\{G_i\}_{i=1}^{\s{dim \CC}}$ be an ONB for $\CC$. For each $j \in \{ 1, 2, \cdots, dim \CC\}$, define $d+1 \times d+1$ real antisymmetric matrix $\Gamma_j$, whose matrix elements are given by $(\Gamma_j)_{kl} = iTr( G_j [T_k,T_l])$. Let $\mathcal{G}$ be the real vector space spanned by the $\Gamma_j$'s. Let $\{ \Omega_j \}_{j=1}^{\s{dim\mathcal{G}}}$ be an ONB for $\mathcal{G}$. Theorem \ref{thm2} allows us to assume that $dim \CC \le d-1$, which implies $dim \mathcal{G} \le d-1$.

\begin{thm}
\label{thm3}
When $dim \Tb = d+1$, $\Tb$ contains a MAS if and only if $\Omega_j$ is rank $2$ for all $j=1,2,\cdots, dim\mathcal{G}$ and $\cap_{j=1}^{\s{dim\mathcal{G}}} Supp ( \Omega_j )$ is one dimensional. 
\end{thm}
\begin{proof} \textbf{IF} Assume that $\Omega_j$ is rank $2$, $\forall$ $1 \leq j \leq dim\mathcal{G}$, and $\cap_{j=1}^{\s{dim\mathcal{G}}} Supp ( \Omega_j )$ is one dimensional, spanned by the real $(d+1)$-tuple $\un{e}_{\s{d+1}}$ $\equiv$ $(e_{\s{1 \, d+1}},e_{\s{2 \, d+1}},\cdots,e_{\s{d+1 \, d+1}})^{T}$. Since $\Omega_j$ is anti-symmetric and real, and since it is rank $2$, there exists a real $(d+1)$-tuple $\un{e}_{\s{j}}$ $\equiv$ $(e_{\s{1j}},e_{\s{2j}},\cdots,e_{\s{ d+1 \, j}})^{T}$ so that $\Omega_j = \un{e}_{\s{d+1}} {\un{e}}_{\s{j}}^{\s{T}} - \un{e}_{\s{j}} {\un{e}_{\s{d+1}}}^{\s{T}}$. Since $\Omega_j$ is invariant for any arbitrary value of inner product ${\un{e}_{\s{j}}}^{\s{T}}.\un{e}_{\s{d+1}}$, choose $\un{e}_{\s{j}}$ to be orthogonal to $\un{e}_{\s{d+1}}$. Then, for $\Omega_j$ to be orthogonal to $\Omega_{j'}$, it is required that ${\un{e}_{\s{j}}}^{\s{T}}. \un{e}_{\s{j'}}=0$. Let $\Gamma_j = \sum_{k=1}^{\s{dim\mathcal{G}}} \alpha_{kj} \Omega_{k}$ $=$ $\un{e}_{\s{d+1}} {\un{g}}_{\s{j}}^{\s{
T}} - \un{g}_{\s{j}} {\un{e}_{\s{d+1}}}^{\s{T}}$, where $\un{g}_{\s{j}}$ $\equiv$ $\sum_{k=1}^{\s{dim \mathcal{G}}} \alpha_{kj} \un{e}_{\s{k}} $. Hence $\Gamma_j$ are also rank $2$ matrices. Complete the ONB $\{ \un{e}_{\s{1}},\un{e}_{\s{2}},\cdots, \un{e}_{\s{dim \mathcal{G}}}$ $\un{e}_{\s{dim \mathcal{G}+1}}, \cdots, \un{e}_{\s{d+1}} \}$. One can normalize $\Omega_j$ so that $\{ \un{e}_{\s{j}} \}_{j=1}^{d+1}$ is an ONB for $\mathbb{C}^{d+1}$. Arrange ${\un{e}_{\s{j}}}^{\s{T}}$ as rows of a $d+1 \times d+1$ orthogonal matrix $O$ in ascending order of $j$ from $0$ to $d+1$. Then $O \, \un{e}_{\s{1}} = (1,0,0,\cdots,0)^T$, $O \, \un{e}_{\s{2}} = (0,1,0,\cdots,0)^T$, $\cdots$, $O \, \un{e}_{\s{d+1}} = (0,0,0,\cdots,1)^T$. For all $ 1 \le j \le dim \CC$ and for $1 \le k , l \le d$ \small
\begin{equation}
\label{if} 
\left( O \Gamma_j O^{\s{T}}\right)_{kl} = 0 \Longrightarrow Tr\left( G_j [T'_k,T'_l] \right)   = 0,
\end{equation}
where $T'_k \equiv \sum_{l=1}^{d+1}O_{kl}T_l$. Since $T'_k \in \Tb$, $[T'_k,T'_l]$ $\in$ $\CC$. But since $\{G_j\}_{j=1}^{\s{dim \CC}}$ is an ONB for $\CC$, equation \eqref{if} implies that $[T'_k,T'_l] = 0$. Thus $\{ T'_j \}_{j=1}^{d}$ spans a MAS in $\Tb$. \textbf{ONLY IF} Assume $\Tb$ contains a MAS and let $\{ T_j \}_{j=1}^{d}$ be an ONB for this MAS. Then $Tr(G_j[T_k,T_l]) = 0$ when $1 \leq k , l \leq d$. Thus the $d \times d$ upper diagonal block of $\Gamma_j$ is zero, which makes it a rank $2$ matrix. The same is true for $\{\Omega_j\}_{j=1}^{dim \mathcal{G}}$. For $\Omega_j$ and $\Omega_{j'}$ to be orthogonal one requires their corresponding $d+1$-th columns (and $d+1$-th rows) to be orthogonal as well. This implies $\cap_{j=1}^{\s{dim\mathcal{G}}} Supp ( \Omega_j )$ is spanned by only one vector which is $(0,0, \cdots ,0,1)^T$.
\end{proof}

Next, I give an example of theorem \eqref{thm3}.

\begin{example}
\label{exam3}
\label{ex3}
For \small$\{ \ket{\psi_{00}}_{\s{AB}}$, $\ket{\psi_{01}}_{\s{AB}}$, $\ket{\psi_{12}}_{\s{AB}}$, $\ket{\psi_{30}}_{\s{AB}}\}$\normalsize, $\Tb$ is spanned by $\{T_1=  \mathbb{1}_4$, $T_2=W_{02}$, $T_3=\frac{W_{21}-W_{23}}{2}$, $T_4=\frac{W_{21}+W_{23}}{2i}$, $T_5=W_{20} \}$. $\CC$ is spanned by \small \begin{align}
\label{2ex4}
G_1 = \frac{1}{2\sqrt{2}}\begin{pmatrix}
            0 &  1 & 0 & 1\\
            -1 &  0 & -1 & 0\\
            0 &  1 & 0 & 1\\
            -1 &  0 & -1 & 0
            \end{pmatrix}, 
G_2 = \frac{i}{2\sqrt{2}}\begin{pmatrix}
            0 & 1  & 0 & 1\\
            1 &  0 & 1 & 0\\
            0 &  1 & 0 & 1\\
            1 &  0 & 1 & 0
            \end{pmatrix}.
\end{align}\normalsize 
$\mathcal{G}$ is spanned by \small
\begin{align}
\label{2ex5}
\Omega_1 \propto \Gamma_1  = \begin{pmatrix}
            0 &  0 & 0 & 0  & 0\\
            0 &  0 & 0 & 0  & 0\\
            0 &  0 & 0 & 0  & 1\\
            0 &  0 & 0 & 0  & 0\\
            0 &  0 & -1 & 0  & 0
           \end{pmatrix}, \Omega_2 \propto \Gamma_2  = \begin{pmatrix}
            0 &  0 & 0 & 0  & 0\\
            0 &  0 & 0 & 0  & 0\\
            0 &  0 & 0 & 0  & 0\\
            0 &  0 & 0 & 0  & 1\\
            0 &  0 & 0 & -1 & 0
           \end{pmatrix}.
\end{align} \normalsize
$\Omega_1$ and $\Omega_2$ are rank $2$ and $ Supp ( \Omega_1 ) \cap  Supp ( \Omega_2 )$ is spanned by $(0,0,0,0,1)^T$. Since $(\Gamma_j)_{kl}$ $=$ $iTr(G_j[T_k,T_l])$, that the $4 \times 4$ upper diagonal block of $\Gamma_j$'s are zero implies that $\{ T_i \}_{i=1}^{4}$ span a MAS in $\Tb$. Upon computing the common eigenbasis of this MAS, we obtained the ONB:$\Big\{\frac{\ket{0}_{\s{A}}+\ket{1}_{\s{A}}+\ket{2}_{\s{A}}+\ket{3}_{\s{A}}}{2} $,  $\frac{ \ket{0}_{\s{A}}-\ket{1}_{\s{A}}-\ket{2}_{\s{A}}+\ket{3}_{\s{A}}}{2} $, $\frac{ \ket{0}_{\s{A}}+\ket{1}_{\s{A}}-\ket{2}_{\s{A}}-\ket{3}_{\s{A}}}{2} $,  $ \frac{\ket{0}_{\s{A}}-\ket{1}_{\s{A}}+\ket{2}_{\s{A}}-\ket{3}_{\s{A}}}{2} \Big\}$. Then the states can be distinguish when Alice starts by measuring in this ONB.
\end{example}

\textbf{Remarks and Summary:} Early in the paper, I made two assumptions to simplify notation. Results derived under these assumptions \emph{actually} hold for the more general scenarios: when either Alice or Bob can start the $1$-LOCC protocol and when $dim \mathcal{H}_{\s{A}}$ and $dim \mathcal{H}_{\s{B}}$ are unequal. A broad summary of results in this paper can then be given as follows: for the $i$-th party of the $d_A \otimes d_B$ dimensional bipartite system, the set of all sets of orthogonal bipartite states can be partitioned into different classes, based on value of $dim \Tb^{(i)}$ of each set of orthogonal bipartite states. In one sweep, results about existence of $1$-LOCC distinguishability protocols, which the $i$-th party can initiate, can be made about all sets of orthogonal bipartite states, which lie in certain classes. To add a final comment on the usefulness of this framework: note 
that in \cite{C08}, Cohen used the \emph{same} structure to show that almost all sets of $\ge d+1$ orthogonal $N$-qudit multipartite states (in $\left(\mathbb{C}\right)^{\otimes N}$) are not distinguishable by LOCC. Hence, I argue that a deeper study of this structure will be a rewarding experience for studying problems of distinguishability of orthogonal states by LOCC.  \smallbreak

\begin{acknowledgments}
I thank the first and the third referees for their useful comments. I also thank my PhD. adviser Sibasish Ghosh for his help.
\end{acknowledgments}


\begin{center}
\textbf{\large Supplemental Material}
\end{center}
\setcounter{equation}{0}
\setcounter{figure}{0}
\setcounter{table}{0}
\setcounter{page}{1}
\makeatletter

\renewcommand{\theequation}{S\arabic{equation}}
\renewcommand{\thefigure}{S\arabic{figure}}
\renewcommand{\bibnumfmt}[1]{[S#1]}
\renewcommand{\citenumfont}[1]{S#1}

\section{Extremality of POVMs}

Consider two POVMs with elements $\{ M^{\s{(1)}}_j\}_{j=1}^{m_1}$ and $\{ M^{\s{(2)}}_j\}_{j=1}^{m_2}$, where $\sum_{j=1}^{m_i} M^{\s{(i)}}_j = \mathbb{1}_{\s{A}}$ (where $\mathbb{1}_{\s{A}}$ is the identity operator acting on $\mathcal{H}_{\s{A}}$), for $i=1,2$. Define $E^{\s{(1)}}$ to be an ordered $(m_1 + m_2)$-tuple of observables acting on $\mathcal{H}_{\s{A}}$, such that the $k$-th component of $E^{\s{(1)}}$, i.e., $E^{\s{(1)}}_{\s{k}}$ is either a POVM element $M^{\s{(1)}}_j$ or is the null observable $0$ acting on $\mathcal{H}_{\s{A}}$, and also let $E^{\s{(1)}}$ be such that each POVM element from $\{ M^{\s{(1)}}_j\}_{j=1}^{m_1}$ appears once as some component of $\EE{1}$. Depending on the arrangement of $M^{\s{(1)}}_j$'s as components of $\EE{1}$, there are $\frac{(m_1 + m_2)!}{m_2!}$ such distinct ordered tuples $E^{\s{(1)}}$ corresponding to the POVM $\{ M^{\s{(1)}}_j\}_{j=1}^{m_1}$. Define $\EE{2}$ similarly for the second POVM. Choosing some two outcome probability $(p,1-p)$, (where $0 \leq p \
leq 1$), one can obtain a new POVM by point wise addition of components of $\EE{1}$ and $\EE{2}$, i.e., the set $\{ p \Ee{1}{k} + (1-p)\Ee{2}{k}, \; \forall \; 1 \leq k \leq m_1 + m_2 \; : \; p \Ee{1}{k} + (1-p)\Ee{2}{k} \neq 0  \}$ contains all elements of a POVM which is obtained from the convex sum of the original POVMs. In this way, the set of all POVMs is a convex set. An extremal POVM in this set is one which cannot be written as a convex sum (in the aforementioned fashion) of two or more distinct POVMs. 

It is also not necessary for a rank-one POVM to be an extremal rank-one POVM. For example, for $d=2$, consider the following POVM elements: $\{\frac{1}{2} \proj{0}, \frac{1}{2} \proj{1}, \frac{1}{2} \proj{+}, \frac{1}{2} \proj{-}\}$, where $\ket{+} \equiv \frac{1}{\sqrt{2}}(\ket{0} + \ket{1})$ and $\ket{-} \equiv \frac{1}{\sqrt{2}}(\ket{0} - \ket{1})$. This POVM is non-extremal because it can be written as a convex sum of two POVMs, $\{ \proj{0},\proj{1}\}$ and $\{ \proj{+}, \proj{-}\}$. The POVMs $\{ \proj{0},\proj{1}\}$ and $\{ \proj{+}, \proj{-}\}$ are extremal because they cannot be written as convex sums of other POVMs. Also, an extremal rank-one POVM need not be a rank-one projective POVM. For instance let $\ket{\tilde{v}_1} \equiv \frac{1}{\sqrt{2}} \ket{0}$, $\ket{\tilde{v}_2}\equiv \frac{1}{2} \ket{0} + \frac{1}{\sqrt{2}} \ket{1}$ and $\ket{\tilde{v}_3} \equiv \frac{1}{2} \ket{0} - \frac{1}{\sqrt{2}} \ket{1}$; then the set $\{ \proj{\tilde{v}_1}, \proj{\tilde{v}_2}, \proj{\tilde{v}_3} \}$ is an 
extremal but non-projective POVM. That said \emph{all} rank-one projective measurements are extremal. A necessary and sufficient condition for extremality of POVMs in terms of the Kraus operators of said measurement was first given by Choi \cite{Choi75}; it can be easily checked that the aforementioned POVM whose elements were $\{ \proj{\tilde{v}_i} \}_{i=1}^{3}$, satisfy these necessary and sufficient conditions for be an extremal POVM. 

\section{An alternative proof to Walgate et al's result \cite{W00}}
\label{wal}

Walgate et al's result \cite{W00} states that any two multipartite orthogonal pure states are always locally distinguishable. Their paper shows that the result for the multipartite case follows straightforwardly from the result for the bipartite case. Their proof for the bipartite case is constructive, i.e., they show that for any two orthogonal bipartite states there exists a $1$-LOCC protocol which Alice and Bob can perform to distinguish the two states. That said their protocol is complicated by the fact that the starting party (assumed here to always be Alice) has to perform SWAPPING operations onto a bigger subsystem. 

This result by Walgate corresponds to the case where $n=2$, i.e., all sets of two orthogonal bipartites states come within the classes corresponding to $dim \Tb \ge d^2 -2$. Here I show that when $dim \Tb \ge d^2 -2$, $\Tb$ always contains a MAS, implying the Alice can initiate the $1$-LOCC protocol by performing an OP rank-one projective measurement. Such a protocol is devoid of requiring any SWAPPING operations onto a bigger system.

\begin{thm}
\label{cor4}
When $dim \Tb \ge d^2 - 2$, $\Tb$ always contains a MAS.
\end{thm}

\begin{proof}
This proof is by induction. Assume that $dim \Tb = d^2 -2$. This implies that  $dim \T =2 $. Let $A$ and $H$ be two linearly independent $d \times d$ matrices in $\T$. \emph{Proposition $P(d)$: For any two $d \times d$ hermitian matrices $H$ and $A$, there exists a $d \times d$ unitary $U$, so that the diagonals of $U^\dag H U$ and $U^\dag A U$ are multiples of $\mathbb{1}_{d}$.} It's known that $P(2)$ is true \cite{W00}. The goal is to prove that $P(d+1)$ is true assuming that $P(d)$ is true. Let $H$ and $A$ be two $d+1 \times d+1$ traceless hermitian matrices. Let $H_d$ and $A_d$ be their $d \times d$ upper diagonal block matrices. Since $P(d)$ is true, there is a $d \times d$ unitary $V_d$, so that diagonals of $V_d^\dag H_d V_d$ and $V_d^\dag A_d V_d$ are multiples of $\mathbb{1}_d$. Embedd $V_d$ as the $d \times d$ upper diagonal block of a $d+1 \times d+1$ unitary $V$ whose $d+1$-th diagonal element is $1$. Then it is easy to see that the diagonals of the $d \times d$ upper diagonal block of $V^\dag H 
V$ and $V^\dag A V$ are scalar multiples of $\mathbb{1}_d$. Since $V^\dag H V$ and $V^\dag A V$ are traceless, their diagonals are scalar multiples of matrix $D_\lambda \equiv \frac{1}{\sqrt{d(d+1)}} Diag(1,1,\cdots,1,-d)$, which is traceless. Let $V^\dag H V$ and $V^\dag A V$ have components $\alpha$ and $\beta$ $\in$ $\mathbb{R}$ along $D_\lambda$. Then $A'\equiv \frac{1}{\sqrt{\alpha^2 + \beta^2}}(-\beta V^\dag H V + \alpha V^\dag A V) $ has a zero diagonal, and component of $D_\lambda$ along $H'\equiv \frac{1}{\sqrt{\alpha^2 + \beta^2}}(\alpha V^\dag H V + \beta V^\dag A V) $ is $1$. Let the $(d,d+1)$-th matrix element of $A'$ be $ae^{-i \phi}$. Define $D_u \equiv Diag(1,1,\cdots,1,e^{\frac{-i(\pi+2\phi)}{4}},e^{\frac{i(\pi+2\phi)}{4}})$, then the $2 \times 2$ lower diagonal block of $A'' \equiv D_u^\dag A' D_u $ is a scalar multiple of $\sigma_y$. The diagonal of $H'' \equiv D_u^\dag H' D_u$ remains invariant. Let the real part of the $(d,d+1)$-th matrix element of $H''$ be $h$. Using an $SO(2)$ 
transformation, rotate between the $d$-th and $d+1$-th matrix elements of $H''$ to obtain $H'''$, while keeping all other elements fixed. $A''$ will remain invariant. Thus the real part of the $2 \times 2$ lower diagonal block of $H'''$ will undergo the transformation\small
\begin{align*}
\begin{pmatrix}
1 & h \\
h & -d
  \end{pmatrix}    
 \longrightarrow 
 \begin{pmatrix}
cos\frac{\theta}{2} & sin\frac{\theta}{2}\\
-sin\frac{\theta}{2} & cos\frac{\theta}{2}
\end{pmatrix} 
\begin{pmatrix}
1 & h \\
h & -d
\end{pmatrix}    
\begin{pmatrix}
cos\frac{\theta}{2} & -sin\frac{\theta}{2}\\
sin\frac{\theta}{2} & cos\frac{\theta}{2}
\end{pmatrix} \\
 = \begin{pmatrix}
         \frac{1-d}{2} + \frac{1+d}{2}  cos\theta + h sin\theta & h cos\theta - \frac{1+d}{2} sin\theta \\ 
         h cos\theta - \frac{1+d}{2} sin\theta & \frac{1-d}{2} - \frac{1+d}{2} cos\theta - h sin\theta
        \end{pmatrix}
\end{align*}\normalsize
I want to solve for $\theta$ in the equation: $\frac{1-d}{2} - \frac{1+d}{2} cos\theta - h sin\theta=0$. When $\theta = 0$, the LHS is $-d$ and when $\theta = \pi$, the LHS is $1$. Since the LHS is a continuous function of $\theta$, there must be some $ \theta \in (0,  \pi)$ for which the LHS is zero. Choose $\theta$ to be this value. Then $H'''$ and $A''$ are matrices whose $d+1$-th diagonal elements are both zero. Using $P(d)$ on the $d \times d$ upper diagonal blocks of $H'''$ and $A''$, $H'''$ and $A''$ can be rotated to obtain corresponding matrices whose diagonals are zero and which span the correspondingly rotated $\mathcal{T}$. Then the correspondingly rotated $\mathcal{T}_\bot$ contains all diagonal matrices which span a MAS.
\end{proof}

\section{Almost All Sets of $d$ Orthogonal Bipartite Pure States in $\h$ correspond to the case $dim \Tb = d$}
\label{app2}
This proof is similar to the Cohen's proof of theorem 1 in \cite{C08}. 

Denote $\G{n,d}$ as the manifold of all sets of $n$ orthogonal bipartite pure states $\{ \ket{\psi_i}_{\s{AB}}\}_{i=1}^{n}$ $\subset$ $\h$, where $\braket{\psi_i}{\psi_{i'}} = \delta_{ii'},$ $\forall$ $1 \leq i < i' \leq n$. Hence every point in $\G{n,d}$ is associated with a set of $d \times d$ orthonormal complex matrices $\{ W_i\}_{i=1}^{n}$ (see equation (2) in main text), i.e.,  $Tr(W_i^\dag W_{i'})=\delta_{ii'},$ $\forall$ $1 \leq i < i' \leq n$. Let's represent the rows of $W_i$ as $\un{w}_{\s{i1}}$, $\un{w}_{\s{i2}}$, $\cdots$, $\un{w}_{\s{id}}$. Vectorize the $W_i$ matrices by arranging these rows $\{\un{w}_{\s{ij}}\}_{j=1}^d$ as complex $d^2$-tuples, i.e., $(\un{w}_{\s{i1}}, \un{w}_{\s{i2}},\cdots,\un{w}_{\s{id}})$ $\in$ $\mathbb{C}^{d^2}$, and arrange these vectorized $W_i$'s as the first upper $n$ rows of a $d^2 \times d^2$ unitary matrix $U$, whose remaining rows are arbitary (insofar as the matrix remains unitary). Hence any point of $\G{n,d}$ can be associated with the first upper $n$ 
columns of a $d^2 \times d^2$ unitary matrix $U$ $\in$ $U(d^2)$. In fact, since the overall phases of these $n$ columns, the permutation of the order of their appearance in the set of first $d$ columns of $U$ and the rest of the $d^2 -n$ columns in $U$ are insignificant to describe the corresponding set of orthogonal pure states from $\h$, the manifold $\G{n,d}$ is given by $U(d^2)/(U(1)^{\times n} \times S_n \times U(d^2-n))$. This is a real manifold.

Let $\ma{d^2}$ be the space of all $d^2 \times d^2$ hermitian matrices, then it is the space of generators for $d^2 \times d^2$ unitary matrices, i.e., if $G \in \mathfrak{u}(d^2)$, then $e^{\s{-iG}}$ is a $d^2 \times d^2$ unitary matrix. Associate the ordered set of the first $n$ rows of $e^{\s{-iG}}$ with the set of $n$ vectorized $W_i$'s. Then the set $\{ W_i \}_{i=1}^{n}$ corresponds to some set of $n$ orthonormal states $\{ \ket{\psi_i}_{\s{AB}} \}_{i=1}^{n}$. This maps any $G$ $\in$ $\mathfrak{u}(d^2)$ to a point in $\G{n,d}$ unambiguously. Let's denote this map by $\mathscr{R}:$ $\mathfrak{u}(d^2)$ $\longrightarrow$ $\G{n,d}$. So $\mathscr{R}(G)$ is a point in $\G{n,d}$ corresponding to $\{ \ket{\psi_i}_{\s{AB}} \}_{i=1}^{n}$. In the following I specify norm-induced-metric for various spaces.
\begin{enumerate}                                                                                                                                                                                                                                                                                                                                                                                                                                                                                                                                                                                                         
\item Metric for all $d^2 \times d^2$ matrices is given by the standard Hilbert Schmidt norm.
\item Let $\{A_i\}_{i=1}^{n}$ be an arbitary set of $n$ complex $d \times d$ matrices, then $|| \{ A_i \}_{i=1}^{n}||$ $=$ $\left(\sum_{i=1}^{n} Tr(A_i^\dag A_i )\right)^\frac{1}{2}$.
\item Let $\{ \ket{\eta_i}_{\s{AB}} \}_{i=1}^{n}$ be a set of $n$ arbitrary vectors in $\h$, then $|| \{\ket{\eta_i}_{\s{AB}}\}_{i=1}^{n}||$ $=$ $(\sum_{i=1}^{n}$ $\prescript{}{\s{AB}}{\braket{\eta_i}{\eta_i}_{\s{AB}}})^\frac{1}{2}$.
\end{enumerate}
Then $G$ $\longrightarrow$ $e^{\s{-i}G}$ is continuous, $e^{\s{-i}G}$ $\longrightarrow$ $\{W_i\}_{i=1}^{n}$ is continuous and $\{ W_i \}_{i=1}^{n}$ $\longrightarrow$ $\{ \ket{\psi_i}_{\s{AB}} \}_{i=1}^{n}$ is continuous. This implies that $\mathscr{R}$ is continuous. It is easy to see that $\mathscr{R}$ is onto but not one-to-one.

For any set of $n$ orthonormal states $\{ \ket{\psi_i}_{\s{AB}} \}_{i=1}^{n}$, one can obtain the $d(d-1)$ matrices $\{H_\mathbf{i}, A_\mathbf{i} \}_{\mathbf{i} \in \mathcal{I}}$. Vectorize each of these matrices and arrange them as rows of a $n(n-1) \times d^2$ matrix $M$. Define $\mathscr{D}: \G{n,d} \longrightarrow \mathbb{R}$ by $\mathscr{D}(\{\ket{\psi_i}_{\s{AB}}\}_{i=1}^{n})$ $\equiv Det(MM^\dag)$. The goal is to establish that for no point in $\G{n,d}$ is there an open neighbourhood $\mathcal{N}$ containing said point such that $\mathscr{D}$ vanishes entirely in $\mathcal{N}$. Since $\mathscr{D}$ is continuous on $\G{n,d}$ and $\mathscr{R}$ is continuous on $\ma{d^2}$, $\mathscr{D} \circ \mathscr{R} $ is continuous on $\ma{d^2}$. Hence, if $\mathscr{D}$ vanishes entirely in some open neighbourhood $\mathcal{N}$ of $\{ \ket{\psi_i}_{\s{AB}} \}_{i=1}^{n}$ in $\G{n,d}$, and if $\mathscr{R}(G)$ $=$ $\{ \ket{\psi_i}_{\s{AB}} \}_{i=1}^{n}$, then there is some open neighrboorhood $\mathfrak{n}$ of $G \in \; 
\ma{d^2} $ where $\mathscr{D} \circ \mathscr{R}$ vanishes entirely too. Hence one needs to show that $\mathscr{D} \circ \mathscr{R}$ doesn't vanish entirely in any open neighrbourhood of any point $G$ in $\ma{d^2}$.

Let $\{ \lambda_i \}_{i=1}^{d^2}$ be an ONB for $\ma{d^2}$. Let $G = \un{\alpha}.\un{\lambda}$ be a point in $\ma{d^2}$ which has an open neighbourhood $\mathfrak{n}$ in which $\mathscr{D}\circ\mathscr{R}$ vanishes entirely. Then there exists some $\epsilon_{\s{s}} \in \mathbb{R}$ be such that $(\un{\alpha} +  \epsilon_s\hat{n}).\un{\lambda}$ $\in$ $\mathfrak{n}$ for all unit vectors $\hat{n}$ lying on $S^{\s{d^2-1}}$. 
\begin{widetext}
Then \begin{equation}
\label{exp}
\begin{split}
e^{-i (\un{\alpha} + \epsilon\hat{n}).\un{\lambda}} = & e^{-i\un{\alpha}.\un{\lambda}} \\ 
+ & \epsilon \; ( -i\hat{n}.\un{\lambda} - \frac{(\hat{n}.\un{\lambda})(\un{\alpha}.\un{\lambda})+(\un{\alpha}.\un{\lambda})(\hat{n}.\un{\lambda})}{2!} + i \frac{(\un{\alpha}.\un{\lambda})^2(\hat{n}.\un{\lambda})+(\un{\alpha}.\un{\lambda})(\hat{n}.\un{\lambda})(\un{\alpha}.\un{\lambda}) + (\hat{n}.\un{\lambda})(\un{\alpha}.\un{\lambda})^2}{3!} + \cdots)\\
+ & \epsilon^{\s{2}} \; ( - \frac{(\hat{n}.\un{\lambda})^2}{2} + i  \frac{(\hat{n}.\un{\lambda})^2(\un{\alpha}.\un{\lambda})+(\hat{n}.\un{\lambda})(\un{\alpha}.\un{\lambda})(\hat{n}.\un{\lambda}) + (\un{\alpha}.\un{\lambda})(\hat{n}.\un{\lambda})^2}{3!} + \cdots) \\
+ & \mathcal{O}( \epsilon^{\s{3}}).
\end{split}                                                                                                                                                                                                                                                                                                                                                                                                                                                                                                                                                                                                                          \end{equation}
\end{widetext} 

Hence it is easy to see that as $G$ $\longrightarrow$ $G + \epsilon\hat{n}.\un{\lambda}$, the $W_i$ matrices transform as $W_i$ $\longrightarrow$ $ W_i + \epsilon W_i^{\s{(1)}}(\hat{n}) + \epsilon^{\s{2}} W_i^{\s{(2)}}(\hat{n}) + \mathcal{O}( \epsilon^{\s{3}})$, where $ \epsilon W_i^{\s{(1)}}(\hat{n})$ is the first order change in $ \epsilon$, $ \epsilon^{\s{2}} W_i^{\s{(2)}}(\hat{n})$ is the second order change in $ \epsilon$ and so on. Since equation \eqref{exp} gives the Taylor series expansion of $e^{\s{-i (\un{\alpha} + \epsilon\hat{n}).\un{\lambda}}}$ about $\epsilon=0$, $ W_i + \sum_{k=1}^{\infty} \epsilon^{\s{k}} W_i^{\s{(k)}}(\hat{n})$ is the Taylor series expansion of about $\epsilon=0$. In fact the radius of convergence for the latter is determined by the former, and since the expression in \eqref{exp} converges for all $\epsilon \in \mathbb{R}$ for the former, it does so too for the latter. Now $\mathscr{D}(\{ \ket{\psi_i}_{\s{AB}}\}_{i=1}^{n}) \equiv Det(MM^\dag)$ is a polynomial of the matrix 
elements of $W_i$. So when $W_i$ goes to $W_i + \sum_{k=1}^{\infty} \epsilon^{\s{k}} W_i^{\s{(k)}}(\hat{n})$,  $(\mathscr{D}\circ\mathscr{R})(G)$ $\longrightarrow$ $(\mathscr{D}\circ \mathscr{R})(G)$ $+$ $ \epsilon (\mathscr{D}\circ\mathscr{R})^{\s{(1)}}(\hat{n})$ $+$ $ \epsilon^{\s{2}} (\mathscr{D}\circ \mathscr{R})^{\s{(2)}}(\hat{n})$ $+$ $\mathcal{O}( \epsilon^{\s{3}})$, where $ \epsilon (\mathscr{D}\circ \mathscr{R})^{\s{(1)}}(\hat{n})$ is the first order change in $ \epsilon$, $ \epsilon^{\s{2}} (\mathscr{D}\circ \mathscr{R})^{\s{(2)}}(\hat{n})$ is the second order change in $ \epsilon$ and so on. Note that $(\mathscr{D}\circ\mathscr{R})(G)$ $+$ $\sum_{k=1}^{\infty} \epsilon^{\s{k}} (\mathscr{D}\circ \mathscr{R})^{\s{(k)}}(\hat{n})$ is the Taylor series of $\mathscr{D}\circ \mathscr{R}$ about $G$ in the direction $\hat{n}$. Since the Taylor series  $ W_i + \sum_{k=1}^{\infty} \epsilon^{\s{k}} W_i^{\s{(k)}}(\hat{n})$ convergences for all $\epsilon \in \mathbb{R}$, and since $\mathscr{D}$ is a polynomial 
in the matrix elements of $W_i$, the radius of convergence for the Taylor expansion $(\mathscr{D}\circ\mathscr{R})(G)$ $+$ $\sum_{k=1}^{\infty} \epsilon^{\s{k}} (\mathscr{D}\circ \mathscr{R})^{\s{(k)}}(\hat{n})$ is $\epsilon = \infty$. 

Now let $\mathscr{D}\circ \mathscr{R}$ vanish in $\mathfrak{n}$. This implies that $(\mathscr{D}\circ\mathscr{R}) (G + \epsilon \hat{n})=0$, for all $\hat{n} \in S^{\s{d^2-1}}$ and $\epsilon \in [0,\epsilon_s]$, where $\epsilon_s$ was chosen so that $(\un{\alpha} +  \epsilon_s\hat{n}).\un{\lambda}$ $\in$ $\mathfrak{n}$. The Taylor series of $\mathscr{D}\circ\mathscr{R}$ about $G$ is a summation of monomials in $\epsilon$, i.e., $(\mathscr{D}\circ\mathscr{R})^{(k)}\epsilon^{\s{k}}$, which are linearly independent in the range $\epsilon \in [0,\epsilon_{\s{s}}]$. Hence the only way that such a summation vanishes for all $\epsilon \in [0,\epsilon_{\s{s}}]$ is if $\mathscr{D}^{\s{(k)}}(\hat{n}) = 0$ for all $k \in \mathbb{N}$ and $\hat{n} \in S^{\s{d^2-1}}$, and if $(\mathscr{D}\circ\mathscr{R})(G)=0$. But note that the radius of convergence for $\epsilon$ in this Taylor series is $\infty$. Hence  $\mathscr{D}\circ \mathscr{R}$ vanishes all over $\ma{d^2}$. And that implies that $\mathscr{D}$ vanishes all over 
$\G{n,d}$. The  following counter-example will disprove this: let $ \ket{\psi_i}_{\s{AB}} \equiv \ket{s_i}_{\s{A}} \ket{0}_{\s{B}}$, where $\ket{0}_{\s{B}}\in \mathcal{H}_B$. Then $ Tr_B ( \ketbra{\psi_{i}}{\psi_{i'}} ) =  \ketbra{s_{i}}{s_{i'}}$ when $i \neq i'$, so $\T$ is spanned by the complex congugate of matrices representing $\frac{1}{2} \left( \ketbra{s_{i}}{s_{i'}} + \ketbra{s_{i'}}{s_{i}} \right)$ and $\frac{1}{2i} \left( \ketbra{s_{i}}{s_{i'}} - \ketbra{s_{i'}}{s_{i}} \right)$, for all $1 \leq i < i' \leq d$,  in the standard basis. All these matrices are linearly independent, so $dim \Tb = d$ and $\mathscr{D}(\{ \ket{s_{i}}_{\s{A}} \ket{0}_{\s{B}} \}_{i=1}^{n})$ $\neq$ $0$. Hence it is not possible for $\mathscr{D}$ to vanish entirely in any open neighbourhood of any point in $\G{n,d}$. This also holds true for the particular case when $n=d$. 

\bibliographystyle{apsrev4-1}
\bibliography{mybib}
\end{document}